\documentclass[12pt]{article}

\usepackage{fullpage,float,amssymb,amsmath,amsthm,amscd}
\usepackage{epsf}
\usepackage{float}

\usepackage[usenames]{color}
\usepackage[colorlinks=true,
linkcolor=webgreen,
filecolor=webbrown,
citecolor=webgreen]{hyperref}

\definecolor{webgreen}{rgb}{0,.5,0}
\definecolor{webbrown}{rgb}{.6,0,0}

\usepackage{graphicx}

\def\modd#1 #2{#1\ \mbox{\rm (mod}\ #2\mbox{\rm )}}

\author{Jason Bell\\
Department of Pure Mathematics \\
University of Waterloo \\
Waterloo, ON N2L 3G1 \\
Canada\\ 
{\tt jpbell@uwaterloo.ca} \\
\and Thomas Finn Lidbetter and Jeffrey Shallit\\
School of Computer Science\\
University of Waterloo \\
Waterloo, ON  N2L 3G1 \\
Canada\\
{\tt finn.lidbetter@uwaterloo.ca} \\
{\tt shallit@uwaterloo.ca}
}

\title{Additive Number Theory via Approximation by Regular Languages}

\begin{document}

\maketitle

\theoremstyle{plain}
\newtheorem{theorem}{Theorem}
\newtheorem{corollary}[theorem]{Corollary}
\newtheorem{lemma}[theorem]{Lemma}
\newtheorem{proposition}[theorem]{Proposition}

\theoremstyle{definition}
\newtheorem{definition}[theorem]{Definition}
\newtheorem{example}[theorem]{Example}
\newtheorem{conjecture}[theorem]{Conjecture}

\theoremstyle{remark}
\newtheorem{remark}[theorem]{Remark}
\newtheorem{observation}{Observation}

\def\Que{{\mathbb{Q}}}
\def\Enn{{\mathbb{N}}}
\def\Zee{{\mathbb{Z}}}
\newcommand{\seqnum}[1]{\href{http://oeis.org/#1}{\underline{#1}}}
\def\mS{{\mathcal{S}}}

\begin{abstract}
We prove some new theorems in additive number theory, using novel techniques from automata theory and formal languages.  As an example of our method, we prove that every natural number $>25$ is the sum of at most three natural numbers whose base-$2$ representation has an equal number of $0$'s and $1$'s.  
\end{abstract}

\section{Introduction}
\label{intro}

Additive number theory is the study of the additive properties of the natural numbers \cite{Nathanson:1996,Nathanson:2000}.   As an example of a theorem in this area, consider Lagrange's theorem:  every natural number is the sum of four squares of natural numbers.  
Suppose $S \subseteq \Enn = \{ 0,1,2,\ldots \}$.  The fundamental problem of additive number theory is to determine whether there exists an integer $m$ such that every element of $\Enn$ (resp., every sufficiently large
element of $\Enn$) is the sum of at most $m$ elements of $S$. 
If so, we call $S$ an {\it additive basis of order $m$} (resp., an {\it asymptotic additive basis of order $m$}).
If such an $m$ exists, we also want to find the smallest such $m$.

Recently there has been interest in solving this problem for sets of integers whose base-$k$ expansions match certain patterns.
For example, Banks \cite{Banks:2016} proved that every natural number is the sum of at most $49$ natural numbers whose base-$10$ expansion is a palindrome.  Next, Cilleruelo, Luca, and Baxter
\cite{Cilleruelo&Luca&Baxter:2017} proved that every natural number is the sum of at most $3$ natural numbers whose base-$k$ expansion is a palindrome, for $k \geq 5$.  Finally, the classification was completed by Rajasekaran, Shallit, and Smith \cite{Rajasekaran&Shallit&Smith:2018},
who proved optimal results for bases $k = 2, 3, 4$.  Their method was to construct an automaton $A$ that accepted the representation of those numbers that are the sum of a certain number of palindromes of certain sizes, and then use a decision procedure to characterize the set of numbers accepted by $A$.

In this paper we introduce a different (but related) automaton-based technique for additive number theory.  Suppose we want to show that a given set $S$ of natural numbers forms an additive basis (resp., asymptotic additive basis) of order $m$.  Instead of considering $S$, we consider a subset $S'$ of $S$ for which the set of base-$k$ representations of its elements forms a regular language.  Such a subset is sometimes called {\it $k$-automatic}; see \cite{Allouche&Shallit:2003}.  For such $S'$, and for natural numbers $m'$, it is known that the language of numbers representable as a sum of $m'$ elements of $S'$ is also $k$-automatic \cite{Bell&Hare&Shallit:2017}.
Then we
show (perhaps with some small number of exceptions that typically can be handled in some other way) that $S'$ forms
an additive basis (resp., asymptotic additive basis) of order $m'$.   
Since $S' \subseteq S$, we have now proved that $m \leq m'$.  
We hope that if $S'$ is chosen appropriately, then in fact $m = m'$.
This is the method of {\it regular underapproximation}.

Analogously, consider a $k$-automatic superset $S''$ of $S$ (that is,
a set for which $S \subseteq S''$).  We then compute the set of numbers {\it not\/} representable as a sum of $m''$ elements of $S''$; this set is also $k$-automatic.    If it is nonempty (resp., infinite),
then $S''$, and hence $S$, cannot be an additive
basis (resp., an asymptotic additive basis) of order $m''$.  In this case we have proved that $m'' < m$.  
We hope that if $S''$ is chosen appropriately, then
$m = m'' + 1$.  This is the method of {\it regular overapproximation}.

We call these two techniques together the {\it method of regular approximation}, and we apply them to a number of different sets that have been previously studied. 
In each case we are able to find the smallest $m$ such that the set forms an additive basis (or asymptotic additive basis) of order $m$.   Although the notion of regular approximation is not new \cite{Mohri&Nederhof:2001,Cordy&Salomaa:2007}, our application of it to additive number theory is.

There is a simple criterion for deciding, given a $k$-automatic set, whether it forms an additive basis of finite order (resp., an asymptotic additive basis).  If it does, there is an algorithm for determining the least $m$
for which it forms an additive basis (resp., an asymptotic additive basis) of order $m$ \cite{Bell&Hare&Shallit:2017}.  The advantage to this approach is that all (or almost all) of the computation amounts to manipulation of automata, and hence can be carried out using existing software tools.   In obtaining our results, we made extensive use of two software packages:  {\tt Grail}, for turning regular expressions into automata \cite{Raymond&Wood:1994}, and {\tt Walnut}, for deciding whether a given $k$-automatic set forms an additive basis of order $m$ \cite{Mousavi:2016} (and more generally,
answering first-order queries about the elements of a $k$-automatic set).

Regular underapproximation does not always give the optimal bound.  For example, define $S_i = \{ n \in \Enn \, : \, n \equiv \modd{i} {3} \}$ for
$0 \leq i \leq 2$, and let $T_n = S_1 \cup \{0,3,6,\ldots, 3n\}$ for $n \geq 0$.   
Then each $T_n$ is a regular underapproximation of $S_0 \, \cup \, S_1$.
However, each $T_n$ forms an additive basis of least order $3$, while
$S_0 \, \cup \, S_1$ forms an additive basis of order $2$.

\section{Notation}

We assume a familiarity with formal languages and automata theory.  For all undefined notions, see, e.g., \cite{Hopcroft&Ullman:1979}.

We define $\Sigma_k = \{ 0,1,\ldots, k-1 \}$.  For $n \in \Enn$ we define
$(n)_k$ to be the canonical base-$k$ representation of $n$ (most-significant digit first, without leading
zeroes).  This is extended to sets $S \subseteq \Enn$ in the obvious
way:  $(S)_k = \{ (n)_k \ : \ n \in S \}$. 
For a word $x \in \Sigma_k^*$ we define $[x]_k$ to be the value of $x$ interpreted
as a number in base $k$ (most significant digit first); this operation is also extended to languages.
For a word
$x \in \Sigma_k^*$, we define $|x|_a$ to be the number of occurrences of the
letter $a$ in $x$.

In this paper, we start by considering six sets and their corresponding languages, defined in Table~\ref{tab1}.
The OEIS column refers to the corresponding entry in the {\it On-Line Encyclopedia of Integer Sequences}
\cite{oeis}.
\begin{table}[H]
\begin{center}
\begin{tabular}{|c|c|c|}
Set & Language & Entry in OEIS \\
\hline
\quad $S_{=} = \{ n \in \Enn \ : \ |(n)_2|_0 = |(n)_2|_1 \, \}$ \quad & \quad$L_{=} = (S_{=})_2$ \quad& \seqnum{A031443} \\[.05in]
\quad $S_{<} = \{ n \in \Enn \ : \ |(n)_2|_0 < |(n)_2|_1 \, \}$ \quad& \quad$L_{<} = (S_{<})_2$\quad & \seqnum{A072600} \\[.05in]
\quad$S_{\leq}= \{ n \in \Enn \ : \ |(n)_2|_0 \leq |(n)_2|_1 \, \}$\quad &\quad $L_{\leq} = (S_{\leq})_2$ \quad& \seqnum{A072601} \\[.05in]
\quad $S_{>} =\{ n \in \Enn \ : \ |(n)_2|_0 > |(n)_2|_1 \, \}$\quad & \quad$L_{>} = (S_{>})_2$ \quad& \seqnum{A072603} \\[.05in]
\quad $S_{\geq} = \{ n \in \Enn \ : \ |(n)_2|_0 \geq |(n)_2|_1 \, \}$\quad &\quad $L_{\geq} = (S_{\geq})_2$\quad & \seqnum{A072602}	\\[.05in]
\quad $S_{\neq} = \{ n \in \Enn \ : \ |(n)_2|_0 \neq |(n)_2|_1 \, \}$\quad &\quad $L_{\neq} = (S_{\neq})_2$\quad & \seqnum{A044951}  \\[.05in]
\hline
\end{tabular}
\end{center}
\caption{Sets considered and their corresponding languages.}\label{tab1}
\end{table}
Note that all these languages are context-free.

When we display DFA's in this paper, any dead state and transitions to the dead
state are typically omitted.

\section{First-order statements and {\tt Walnut}}

We use the free software {\tt Walnut}, which works with first-order formulas and automata.   We fix a base $k$ (usually $k = 2$ in this paper) and assume the
input to the automaton is written in base $k$, starting with the most significant
digit.   We identify an automaton $B$ with
its corresponding characteristic sequence
${\bf b} = (b(n))_{n \geq 0}$, taking the value $1$ if $(n)_k$ is accepted by $B$
and $0$ otherwise.  For technical reasons, $B$ must provide the correct result even
if leading zeroes appear in the input.

{\tt Walnut} can determine whether a first-order statement, including indexing and addition, about the values of $\bf b$ is true.  The underlying domain is the natural numbers.  For example, the first-order statement
\begin{equation}
\forall n \ \exists x,y,z \ (n=x+y+z) \ \wedge \ (b(x) = 1) \ \wedge \ (b(y) = 1) \ \wedge \ (b(z) = 1) 
\label{fir}
\end{equation}
asserts that, for all $n$, the natural number $n$ is the sum of three integers for which the value of $\bf b$ is 
$1$.  In other words, if $S$ is the set of integers whose base-$k$ representation
is accepted by $B$, then $S$ forms an additive basis of order $3$.  

The corresponding {\tt Walnut} command is a straightforward translation of the statement in
\eqref{fir}:

\centerline{\tt A n \ E x,y,z \ (n=x+y+z) \& (B[x]=@1) \& (B[y]=@1) \& (B[z]=@1)  .}

Furthermore, if the prefix {\tt A n} is omitted, then the result is an automaton that
accepts the base-$k$ representations of those $n$ having a representation as a sum of $3$ elements of $S$.    A similar first-order statement makes the analogous assertion with ``all $n$" replaced by ``all sufficiently large $n$".

\section{Example of the method:  the set $S_{\geq}$}

We start with a very simple example of our method, discussing the additive properties of those numbers with at least as many $0$'s as $1$'s in their base-$2$ expansion.
The first few such numbers are 
$$ 2,4,8,9,10,12,16,17,18,20,24,32,33,34,35,36,37,38,40, \ldots  .$$

\begin{theorem}
Every natural number except $1$, $3$, $5$, $7$ is the sum of at most three elements of $S_{\geq}$.
\label{geq-thm}
\end{theorem}

We start with a ``conventional'' proof of this theorem.  The reader will note the use of an argument requiring several special cases.

\begin{proof}
Given $N$ we want to represent, let $n_1$ (resp., $n_2$, $n_3$) be the integer
formed by taking every $3$rd $1$, starting with the
first $1$ (resp., second $1$, third $1$), in the base-$2$ representation of $N$.
Provided there are at least four $1$'s in $(N)_2$,
every $1$ in $(n_i)_2$, for $1 \leq i \leq 3$, is associated with at least two following zeroes, except possibly the very last $1$, and hence $n_i \in S_{\geq}$.

This construction can fail on odd numbers whose base-$2$ representation has three $1$'s or fewer, so we must treat those as special cases.

For numbers of the form $N = 2^i + 1$ with $i \geq 3$, we can
take $n_1 = 2^i + 1$, $n_2 = n_3 = 0$.

For numbers with binary representation $1 0^i 1 0^j 1$,
we can take $n_1 = [1 0^{i+j+1} 1]_2$,
$n_2 = [1 0^{j+1}]_2$, $n_3 = 0$.   This works provided
$i + j + 1 \geq 2$ and $j+1 \geq 1$.  

This covers all cases except $N = 1,3,5,7$.   
\end{proof}

Now we reprove the same theorem, using our method of regular approximation.
We start by finding a regular language that is both (a) sufficiently dense and for which the represented numbers form
(b) a subset of $S_{\geq}$.  After a bit of experimentation, we choose
$L_1 = 1(01+0)^*-(10)^*1 = (10)(10)^*(0(0+10)^*(1 + \epsilon) + \epsilon)$.

\begin{theorem}
Every natural number except $1$, $3$, $5$, $7$ is the sum of at most three natural numbers whose base-$2$
representations lie in the regular language
$L_1 = 1(01+0)^*-(10)^*1 = (10)(10)^*(0(0+10)^*(1+\epsilon)+\epsilon)$.
\label{geqa-thm}
\end{theorem}
\begin{proof}
First, use the {\tt Grail} command

{\tt echo '0*+0*10(10)*(0(0+10)*(1+"")+"")' | retofm | fmdeterm |}\\
{\tt fmmin | fmcomp | fmrenum > ge1}

\noindent to create an automaton {\tt ge1} accepting $L_1$. 
Here {\tt ""} is {\tt Grail}'s way of
representing the empty word $\epsilon$. Note that
every element of $L_1$ has at least
as many $0$'s as $1$'s.  
Also, we added $0^*$ in two places to get all representations with leading zeroes, including all representations of $0$.
This produces the automaton given in
Figure~\ref{fig10} below.
\begin{figure}[H]
\begin{center}
\includegraphics[width=5in]{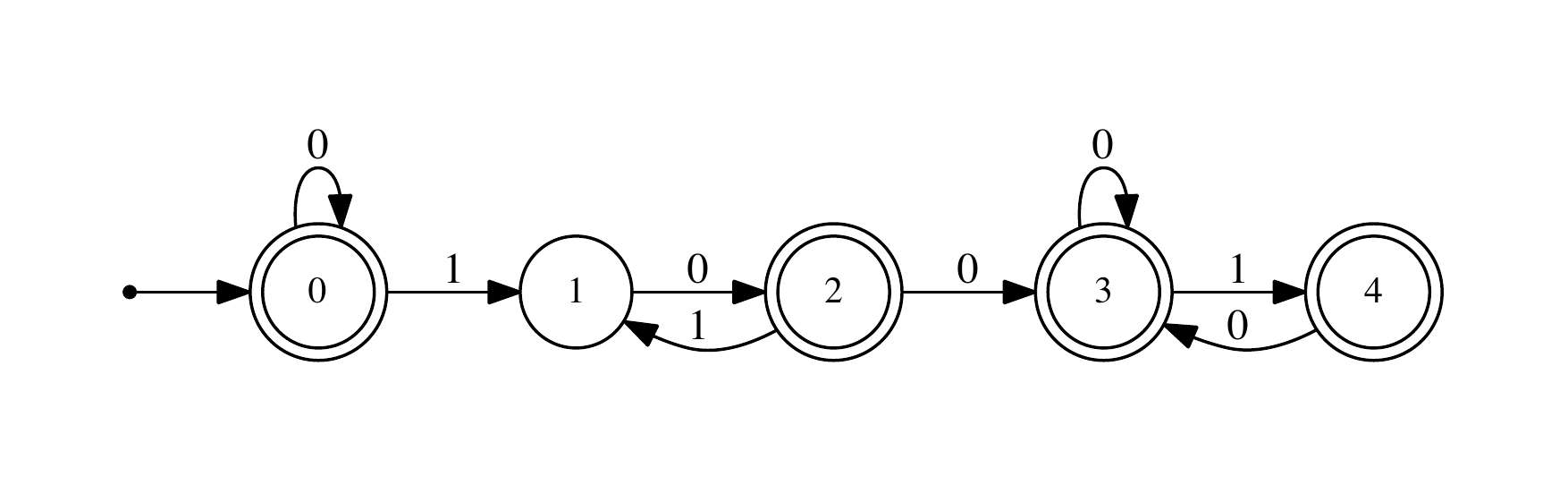}
\end{center}
\caption{Automaton for $L$}
\label{fig10}
\end{figure}
\noindent Next, we create the corresponding automaton {\tt GE} in
{\tt Walnut}, and we use the {\tt Walnut} command

\centerline{\tt eval geq "E x,y,z (n=x+y+z)\&(GE[x]=@1)\&(GE[y]=@1)\&(GE[z]=@1)":}

\noindent giving us the automaton, in Figure~\ref{fig11}, accepting the representations of numbers that are the sum of three elements whose representations are in $L_1$.
\begin{figure}[H]
\begin{center}
\includegraphics[width=5in]{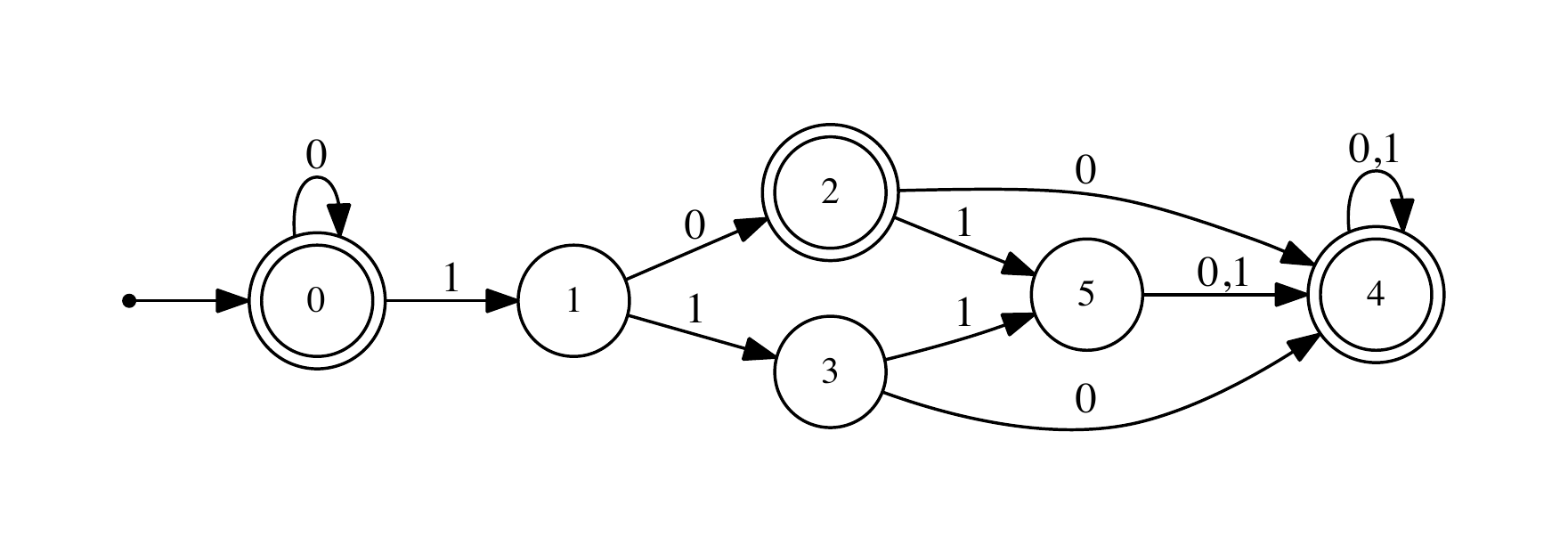}
\end{center}
\caption{Numbers having representations as sums of at most three numbers with representations in $L_1$}
\label{fig11}
\end{figure}
By inspection we easily see that this latter automaton accepts the base-$2$ representation of all numbers except $1$, $3$, $5$, $7$.
Furthermore, it is easy to check that none of $1,3,5,7$
have a representation as a sum of $3$ members of $S_{\geq}$.  This completes the proof of Theorem~\ref{geqa-thm}, which immediately implies Theorem~\ref{geq-thm}.   
\end{proof}

We now show that the bound of $3$ is optimal.
\begin{theorem}
The set $S_{\geq}$ does not form an asymptotic additive basis of order $2$.
\end{theorem}
\begin{proof}
We prove that numbers of the form $2^n - 1$, $n \geq 1$, have no representation
as sums of one or two elements of $S_{\geq}$.  For one element it is clear.
Suppose $2^n - 1 = x + y$ where $x, y \in S_{\geq}$.  If both $x$ and $y$
have less than $n-1$ bits, then their sum is at most $2^n -2$, a contradiction.
Similarly, if both $x$ and $y$ have $n$ bits, then their sum is at least
$2^n$, a contradiction.  So without loss of generality $x$ has $n$ bits
and $y$ has $t < n$ bits.  
Since $(x)_2 \not\in 1^+$,
we can write $(x)_2 = 1^i 0 t$ for $i \geq 1$ and
some word $t$ of length $j = n-i-1$.  Then $(y)_2 = 1 \overline{t}$.
Since $y \in S_{\geq}$ we must have that
$\overline{t}$ contains at least $(j+1)/2$ zeroes.  Then $t$ contains
at most $(j-1)/2$ zeroes.  Then $(x)_2$ contains at most $(j+1)/2 \leq (n-i)/2$ zeroes.  Since $i \geq 1$, this shows $x \not\in S_{\geq}$, a contradiction.   
\end{proof}
One advantage to our method of approximation by regular languages is that it can work in cases where a conventional argument is rather complicated, as in the next section.  Furthermore, the method also gives an $O(\log n)$-time algorithm to find a representation of any given $n$ as a sum of terms of the set, although the implied constant can be rather large.
\begin{remark}
We can also prove that Theorem~\ref{geq-thm} holds even when the summands are required to be
distinct.  We can prove this using the {\tt Walnut} command

{\tt eval geq2 "E x,y,z ((n=x)|(n=x+y)|(n=x+y+z))\&(x<y)\&(y<z)\&}\\ \vphantom{X}\quad\quad\quad\quad\quad {\tt (GE[x]=@1)\&(GE[y]=@1)\&(GE[z]=@1)":}
\end{remark}

\section{The set $S_{=}$}

In this section, we discuss those numbers having an equal number of $0$'s and $1$'s in their base-$2$ expansion.  Such numbers are sometimes called ``digitally balanced".  

\begin{theorem}
Every natural number, except $1$, $3$, $5$, $7$, $8$, $15$, $17$, $25$,
is the sum of
at most three elements of $S_{=}$.
\label{eq-thm}
\end{theorem}
To prove this we prove the following stronger result.
\begin{theorem}
Every natural number,  except 1, 3, 5, 7, 8, 15, 17, 25, 67,  is the sum of
at most three natural numbers whose base-$2$ representations lie
in the regular language
$L_2 = 10(10+01+1100+0011)^* + 1(10+01)^*0 + \epsilon$.
\end{theorem}

\begin{proof}
We used the {\tt Grail} command

{\tt echo '0*10(10+01+1100+0011)*+0*1(10+01)*0+0*' | retofm |}\\
{\tt fmdeterm | fmmin | fmcomp | fmrenum > e1}

\noindent to find the 16-state automaton below in Figure~\ref{fig3}.
\begin{figure}[H]
\begin{center}
\includegraphics[width=6.5in]{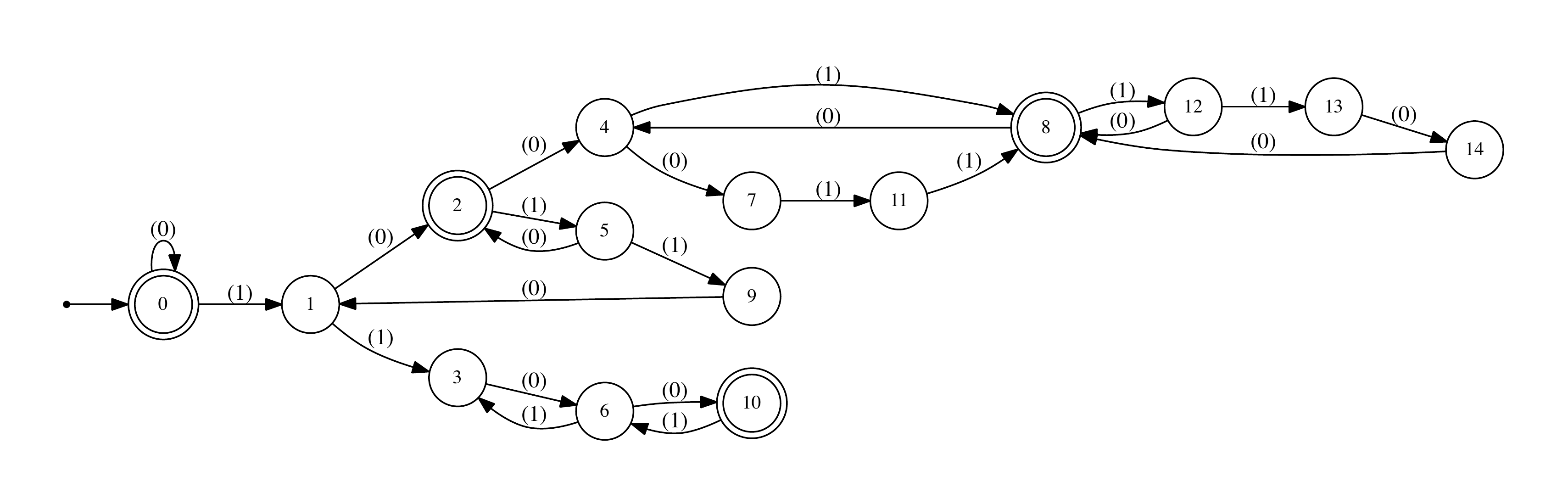}
\end{center}
\caption{Automaton for $0^*10(10+01+1100+0011)^*+0^*1(10+01)^*0+0^*$ }
\label{fig3}
\end{figure}

We then built the corresponding automatic sequence {\tt QQ} in {\tt Walnut}
and issued the command

\centerline{\tt eval eqq "E x,y,z (n=x+y+z)\&(QQ[x]=@1)\&(QQ[y]=@1)\&(QQ[z]=@1)":}

\noindent which produced the 12-state automaton in Figure~\ref{fig4}.  
The total amount of computation time here was 226497 ms, and involved the determinization of an NFA of 1790 states, so this was quite a nontrivial computation for {\tt Walnut}.  By inspection we easily see that
this automaton accepts the base-$2$ representations of all integers except
1, 3, 5, 7, 8, 15, 17, 25, 67.   
\end{proof}

\begin{figure}[H]
\begin{center}
\includegraphics[width=6.5in]{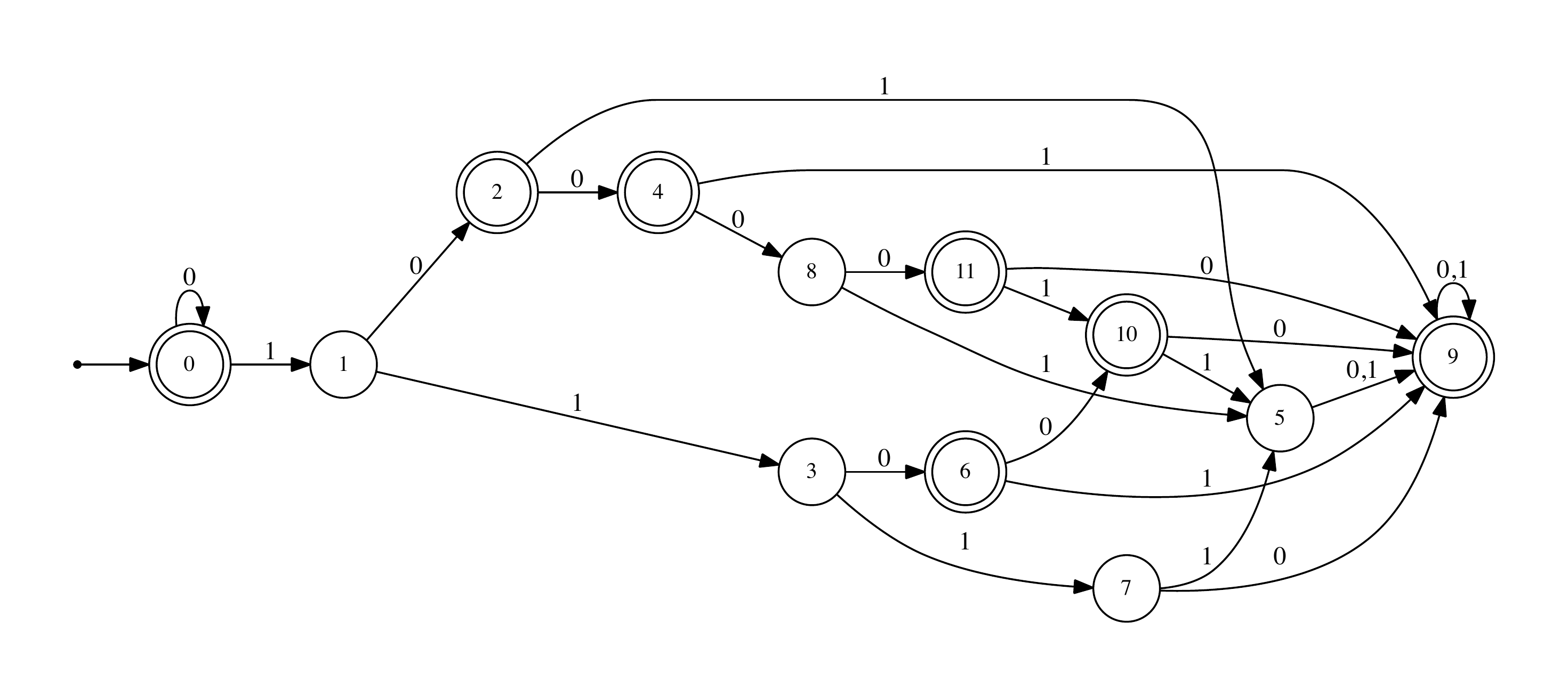}
\end{center}
\caption{Automaton accepting those $n$ that are the sum of at most three numbers with representations in $L_2$}
\label{fig4}
\end{figure}

\begin{proof}[Proof of Theorem~\ref{eq-thm}]
We see that $10(10+01+1100+0011)^*+0^*1(10+01)^*0$ is a regular underapproximation of $L_{=}$.  It is now easy to check that indeed $1,3,5,7,8,15,17,25$ do not have
representations as a sum of at most three members of $S_{=}$, while $67$ has the representation
$67 = 56 + 9 + 2$.  This concludes the proof.   
\end{proof}

We now show that the bound of $3$ is optimal:

\begin{theorem}
There are infinitely many natural numbers that are not the sum of one or two members
of $S_{=}$.
\label{notsum}
\end{theorem}

%\begin{proof}
%Consider $N=2^{2k+1}$ for some integer $k\geq 1$. Observe %that the largest base-2 equinumerous number less than or %equal to $N$ is $2^k(2^{k+1}-1)$, the number with base-2 %representation given by $k$ 1's followed by $k$ 0's. Then %since $2\cdot 2^k(2^k-1)=2^{k+1}(2^k-1)=2^{2k+1}-2^{k+1}%<2^{2k+1}$ we can conclude that $2^{2k+1}$ cannot be %written as the sum of 2 base-2 equinumerous numbers for %\end{proof}

\begin{proof}
We use the method of overapproximation.  Consider
$$S = \{ n \in \Enn \ : \ |(n)_2| \text{ is even but 
$n$ is not of the form $2^k - 1$ } \}. $$
Then it is easy to see that $S_{=} \subset S$.
Furthermore $(S)_2$ is regular, and represented by
the regular language
$$L_3= 1(11)^*(0+100+101)(00+01+10+11)^*.$$

We use {\tt Grail} on the command

{\tt echo '0*+0*1(11)*(0+100+101)(00+01+10+11)*' | retofm | fmdeterm}\\
{\tt | fmmin | fmcomp | fmrenum > ov1}

\noindent giving us the automaton in Figure~\ref{fig23}.
\begin{figure}[H]
\begin{center}
\includegraphics[width=4.5in]{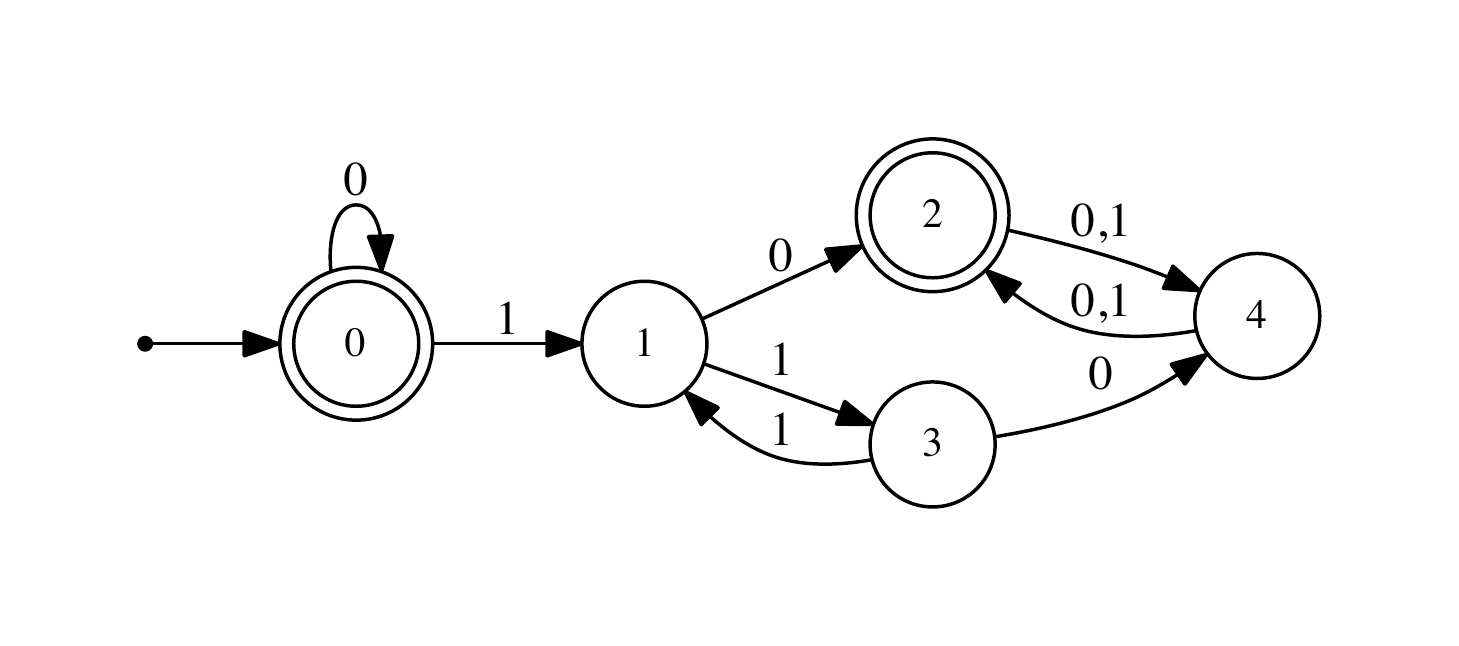}
\end{center}
\caption{Automaton for $L_3$}
\label{fig23}
\end{figure}

Then we ask {\tt Walnut} to give us the base-$2$ representations of all numbers that are {\it not\/} the sum of two members of $S$.  This gives us the automaton in
Figure~\ref{fig24}.
\begin{figure}[H]
\begin{center}
\includegraphics[width=6.5in]{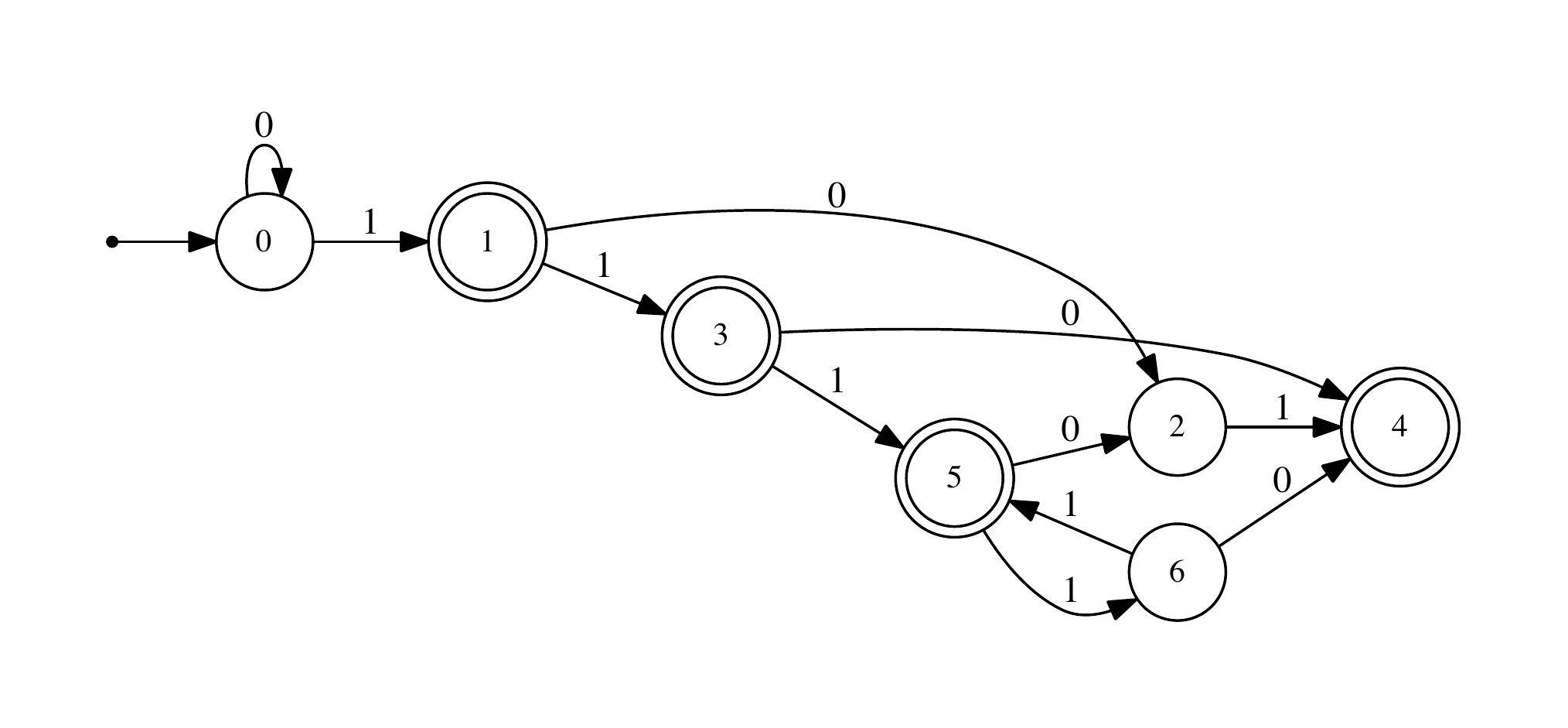}
\end{center}
\caption{Automaton accepting those $n$ that are not the sum of at most two elements of $S$}
\label{fig24}
\end{figure}

By inspection we easily see that numbers with base-$2$
representation $111 (11)^*$ and $111(11)^*0$ have no representation.
Since this set is infinite, we know that $S$ and hence
$S_{=}$ does not form an asymptotic additive basis of order $2$.
 
\end{proof}

\section{The set $S_{<}$}

\begin{theorem}
Every natural number is the sum of at most $2$ elements of $S_{<}$.  The
$2$ is optimal.
\label{lt-thm}
\end{theorem}

\begin{proof}

Our proof is based on the following:

\begin{theorem}
Every natural number is the sum of at most two natural numbers whose base-$2$
representations lie in $1(1+10+01)^*$.
\end{theorem}

We use the {\tt Grail} command

{\tt echo '0*+0*1(1+10+01)*' | retofm | fmdeterm | fmmin | fmcomp |}\\
{\tt fmrenum > lt1}

\noindent which gives the automaton below in Figure~\ref{fig13}.
\begin{figure}[H]
\begin{center}
\includegraphics[width=4.5in]{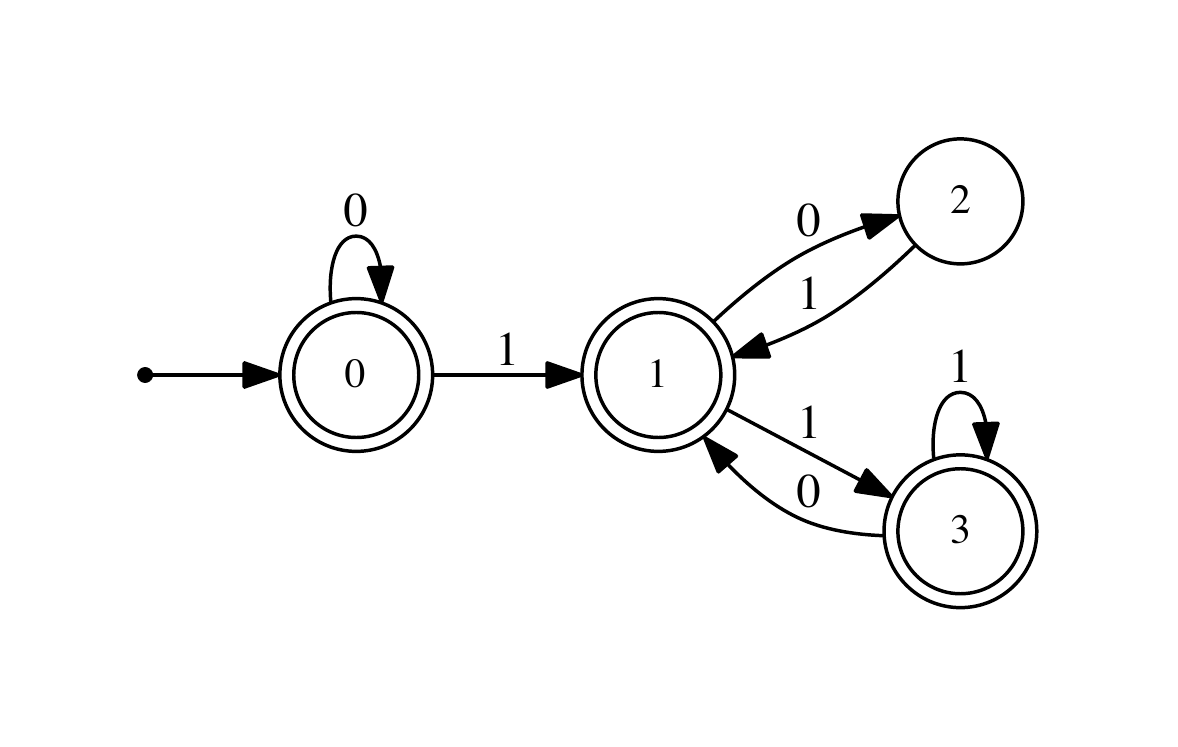}
\end{center}
\caption{Automaton for $1(1+10+01)^*$ }
\label{fig13}
\end{figure}

\noindent Next, we use the {\tt Walnut} command

\centerline{\tt eval lt "E x,y (n=x+y)\&(LT[x]=@1)\&(LT[y]=@1)":}

\noindent which produces a $1$-state automaton accepting everything.  This
concludes the proof. 
\end{proof}

\section{The set $S_{\leq}$}

\begin{theorem}
Every natural number is the sum of at most $2$ elements of $S_{\leq}$.  The
$2$ is optimal.
\label{leq-thm}
\end{theorem}

\begin{proof}
Since $S_{<}$ is a subset of $S_{\leq}$, the result follows immediately. 
\end{proof}

\section{The set $S_{>}$}

\begin{theorem}
Every natural number, except $1$, $2$, $3$, $5$, $6$, $7$, $9$, $10$, $11$, 
$13$, $14$, $15$, $19$, $23$, $27$, $31$, $47$, $63$,
is the sum of at most $3$ elements of $S_{>}$.  The $3$ is optimal.
\label{gt-thm}
\end{theorem}

\begin{proof}

Our proof is based on the following.

\begin{theorem}
Every natural number, except  $1$, $2$, $3$, $5$, $6$, $7$, $9$, $10$, $11$,
$13$, $14$, $15$, $19$, $23$, $27$, $31$, $47$, $63$, $79$,
is the sum of at most $3$ numbers whose base-$2$ representation is given by
$L_4 = 10(01+10+0)^*0(01+10+0)^*$.
\label{gat}
\end{theorem}

\begin{proof}
We use the grail command

{\tt echo '0*+0*10(01+10+0)*0(01+10+0)*' | retofm | fmdeterm | fmmin}\\
{\tt | fmcomp | fmrenum > gt1}

\noindent giving us the automaton in Figure~\ref{fig17}.
\begin{figure}[H]
\begin{center}
\includegraphics[width=6.5in]{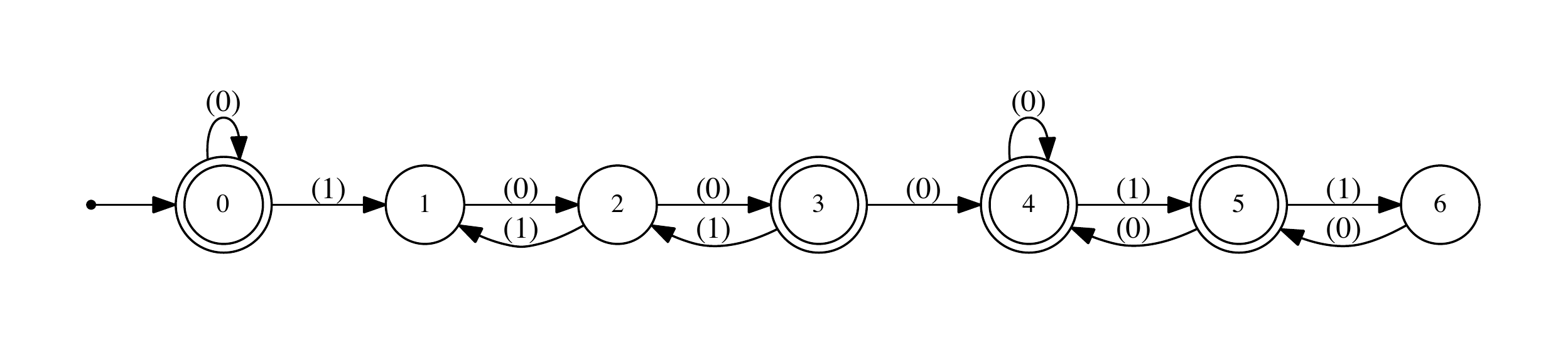}
\end{center}
\caption{Automaton for $10(01+10+0)^*0(01+10+0)^*$}
\label{fig17}
\end{figure}

\noindent Then we use the {\tt Walnut} command

{\tt eval grt "E x,y,z (n=x+y+z)\&(GT[x]=@1)\&(GT[y]=@1)\&(GT[z]=@1)":}

\noindent giving the automaton in Figure~\ref{fig18}.  
\begin{figure}[H]
\begin{center}
\includegraphics[width=6.5in]{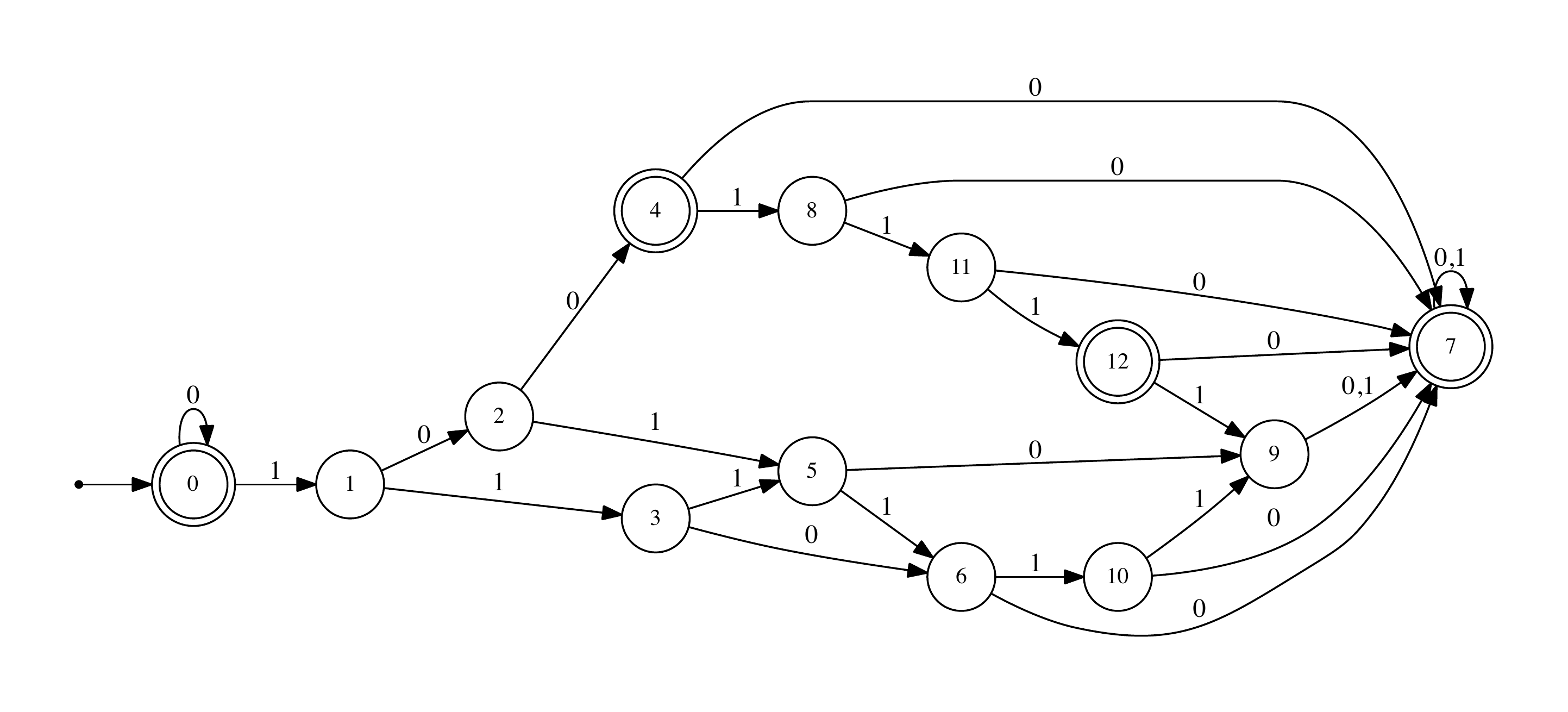}
\end{center}
\caption{Automaton accepting those numbers that
are the sum of at most three elements whose binary representation is in $L_4$}
\label{fig18}
\end{figure}
This concludes the proof.   
\end{proof}

We can now easily check that $1,2,3,5,6,7,9,10,11,13,14,15,19,
23,27,31,47,63$ have no representation as a sum of three elements of $S_{>}$,
while
$79$ has the representation $79=4+8+67$.  So
Theorem~\ref{gat} implies the first claim.  To see that 
two summands do not suffice, note that every element of $S_{>}$ is an
element of $S_{\geq}$, and we already proved above that two summands do not
suffice for $S_{\geq}$. 
\end{proof}

\section{The set $S_{\neq}$}

\begin{theorem}
Every natural number is the sum of at most $2$ elements
of $S_{\neq}$.  The $2$ is optimal.
\label{neq-thm}
\end{theorem}

\begin{proof}
Since $S_{<}$ is a subset of $S_{\neq}$, the result follows immediately. 
 
\end{proof}

\section{The totally balanced numbers}

We say that a word $x \in \{ 0,1 \}^*$ is {\it totally balanced} if 
\begin{itemize}
\item[(a)] $|x|_1 = |x|_0$; and
\item[(b)] $|x'|_1 \geq |x'|_0$ for all prefixes $x'$ of $x$.
\end{itemize}
In other words, such a word is a recoding of a word consisting of balanced parentheses, where $1$ represents
a left parenthesis and $0$ represents a right parenthesis.  The
set of all such words forms a {\it Dyck language} \cite{Chomsky&Schutzenberger:1963}.
Given a totally balanced word $x$, we can define its
{\it nesting level} $\ell(x)$ recursively as follows:
\begin{itemize}
\item[(a)]  $\ell(\epsilon) = 0$;
\item[(b)]  If $x = 1y0z$, where both $y$ and $z$ are totally balanced, then
$\ell(x) = \max(\ell(y)+1,\ell(z))$.  
\end{itemize}

Consider the set of numbers $S_{\rm TB}$ whose base-$2$ representation is totally balanced; note that all such numbers are even.
The elements of $S_{\rm TB}$ form sequence \seqnum{A014486} in the OEIS \cite{oeis}.

\begin{theorem}
Every even number except
$8$, $18$, $28$, $38$, $40$, $82$, $166$ is the sum of at most
$3$ elements of $S_{\rm TB}$.   There are infinitely many even numbers
that are not the sum of at most $2$ elements of $S_{\rm TB}$.
\label{bp-thm}
\end{theorem}

We prove the first part of Theorem~\ref{bp-thm} using the following.

\begin{theorem}
Consider
$$ L_6 = \{ x \in \{0,1\}^* \ : \ x \text{ is totally balanced  and } \ell(x) \leq 3 \} .$$
Then every even
number except $8, 18, 28, 38, 40, 82, 166$ is the sum of at most $3$ natural numbers whose binary representation is
contained in $L_6$.
\end{theorem}

\begin{proof}
The language $L$ is accepted by the following automaton.
\begin{figure}[H]
\begin{center}
\includegraphics[width=5in]{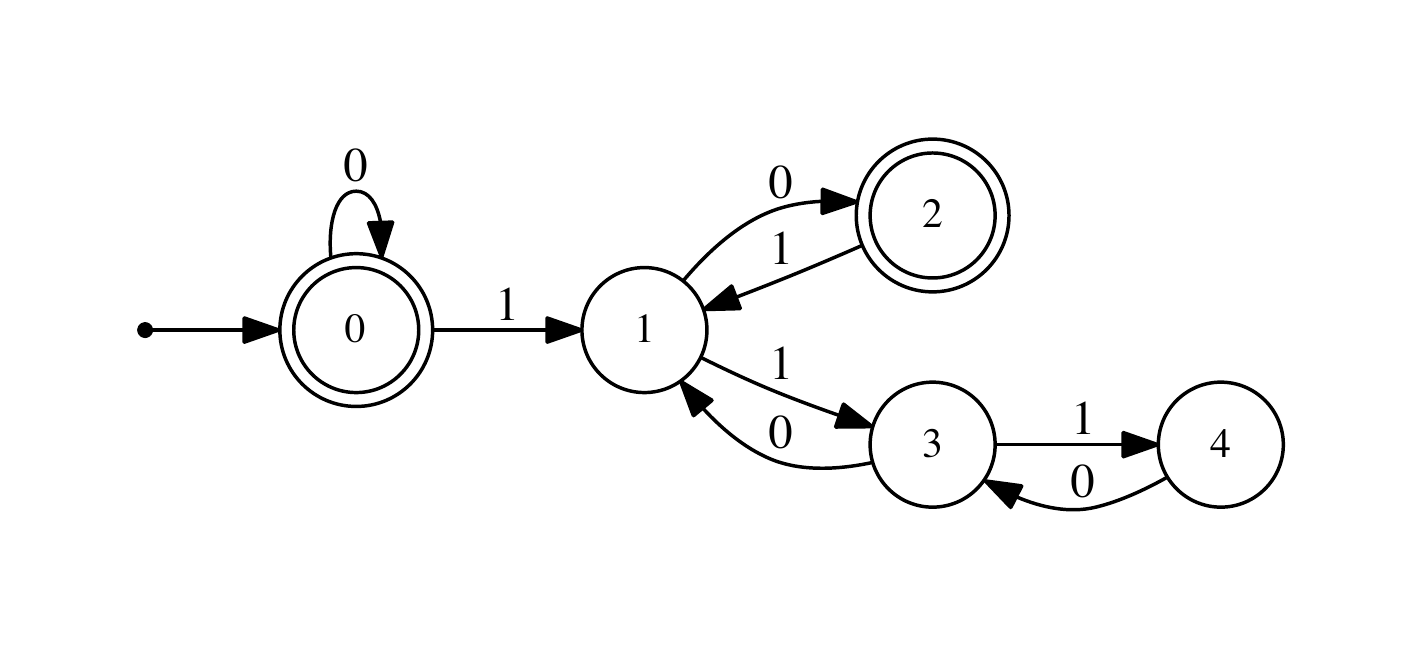}
\end{center}
\caption{Automaton for $L_6$}
\label{fig21}
\end{figure}

Using the {\tt Walnut} command

\centerline{\tt eval 
bp2 "E x,y,z (2*n=x+y+z)\&(TB[x]=@1)\&(TB[y]=@1)\&(TB[z]=@1)":
}

\noindent we get an automaton accepting all $n$ for which $2n$ is
representable.  
\begin{figure}[H]
\begin{center}
\includegraphics[width=6.5in]{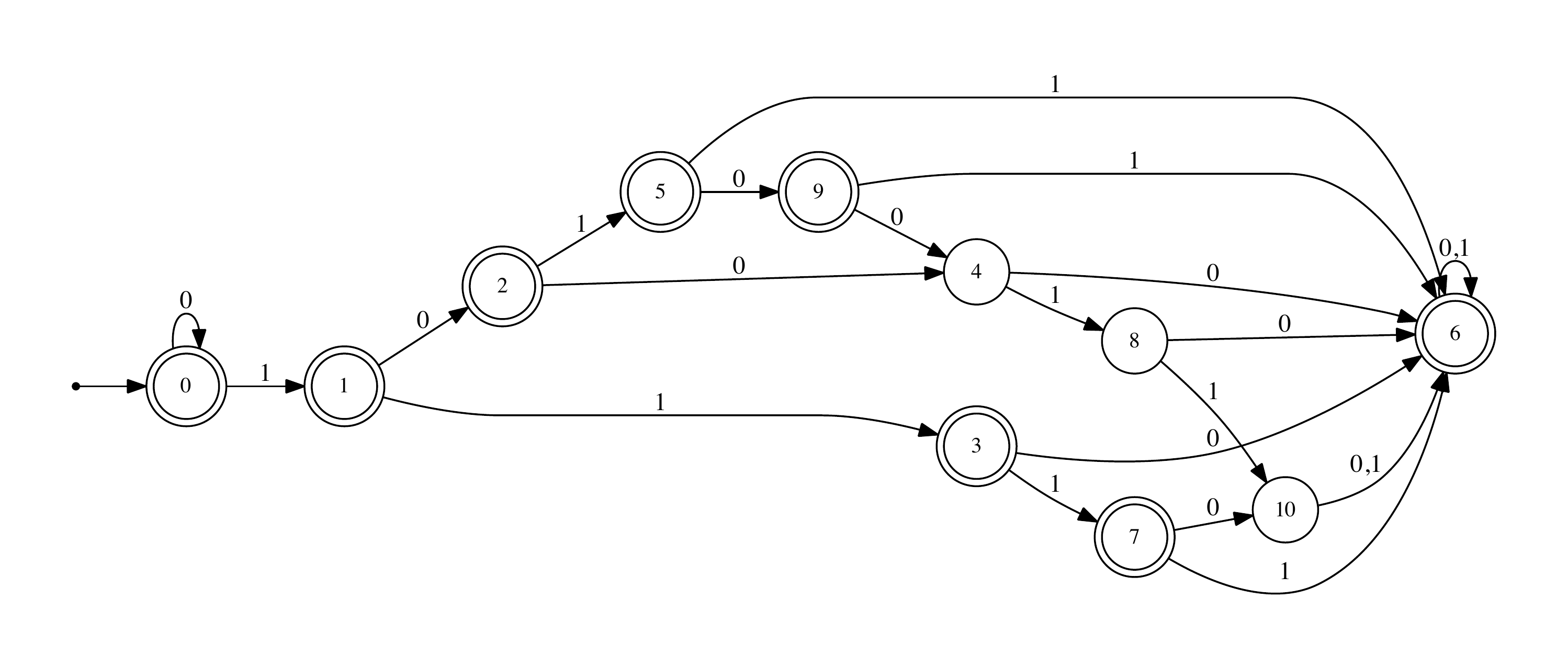}
\end{center}
\caption{Automaton accepting the base-$2$ representation of those $n$ for which $2n = x+y+z+$ with $x,y,z$ having representation in $L_6$}
\label{fig22}
\end{figure}

By inspection it is now easy to verify that the only
$n$ not accepted are $4,9,14,19,20,41,83$. 
\end{proof}

\begin{proof}[Proof of Theorem~\ref{bp-thm}]
Clearly $L_6$ is a subset of $L_{TB}$.   It is easy to check that none of
$8, 18, 28, 38, 40, 82, 166$ have a representation as a sum of three elements of
$S_{TB}$.

To see the second part of the theorem, note that by the proof of Theorem~\ref{notsum} there
are infinitely many even numbers (for example, those with base-$2$ representation
$111(11)^* 0$) not representable as the sum of two elements of $S_{=}$,
and $S_{=}$ is an overapproximation of $S_{TB}$. 
\end{proof}

\section{Generalizations to larger bases}

It would be nice to generalize our results on the languages $L_{=}$, etc., to bases $k > 2$.  However, the appropriate generalization is not completely straightforward, except in the case
of $S_{=}$, the digitally balanced numbers in base $2$.
We can generalize this as follows:
$$ S_{k,=} = \{ n \in \Enn \ : \ |(n)_k|_i = |(n)_k|_j 
	\text{ for all $i, j \in \Sigma_k$} \}.$$

Unfortunately $S_{k,=}$ does not form an additive basis in general, as 
the $\gcd$ of its elements equals $1$ if and only if $k = 2$ or $k = 3$.

\begin{theorem}
The $\gcd$ $g_k$ of the elements of $S_{k,=}$ is
$(k-1)/2$, if $k$ is odd and
$k-1$, if $k$ is even.
\end{theorem}

\begin{proof} 
Let $g_k$ be the $\gcd$ in question.  Consider the two numbers
with base-$k$ representations
\begin{align*}
    & 1 \ 0 \ 2 \ 3 \ \cdots \ (k-3)\ (k-1) \ (k-2) \\
    & 1\  0\  2 \ 3 \ \cdots \ (k-3) \ (k-2)\  (k-1)
\end{align*}
and take their difference.  Now $g_k$ must divide this difference, which
is $k-1$.

Next, take any base-$k$ digitally balanced number $n = \sum_{0 \leq i < t} a_i k^i$
and compute it modulo $k-1$.  We get
    $j(0+1+\cdots+k-1)$
where $j$ is the number of occurrences of each digit.
But this is $j k(k-1)/2$.  It follows that
$$n \equiv \begin{cases}
\modd{j(k-1)/2} {k-1}, & \text{if $k$ is odd}; \\
\modd{0} {k-1}, & \text{if $k$ is even}.
\end{cases}
$$
The result follows. 
\end{proof}

We can now prove

\begin{theorem}
For each $k$ there exists an integer $D = D(k)$ such that every sufficiently large multiple of $g_k$ is a sum of at most $D(k)$ members of $S_{k,=}$.
Furthermore, $D(k) \geq k^{k-1} + 1$.
\label{bases}
\end{theorem}

\begin{proof}
Let $T$ be the set of all words of length $k$ containing each occurrence of
$\Sigma_k$ exactly once, and let $U$ be the words in $T$ that do not begin with
the symbol $0$.  Then $U T^*$ is a regular subset of the language $(S_{k,=})_k$,
and is sufficiently dense that Theorem 1.1 of \cite{Bell&Hare&Shallit:2017} 
applies.

For the lower bound, let $n$ be the smallest element of $S_{k,=}$ of
length $tk$, and let $n'$ be the largest element of $S_{k,=}$ of
length $(t-1)k$.  Then $n-1$ requires at least $\lceil (n-1)/n' \rceil$
summands.  Now $n \geq  k^{tk-1} + 1$ and $n' \leq k^{(t-1)k} - 1$.
It follows that $(n-1)/n' > k^{k-1}$, as desired. 
\end{proof}

With more work we can prove

\begin{theorem}
All natural numbers except 
$1$, $2$, $3$, $4$, $5$, $6$, $7$, $8$, $9$, $10$, $12$, $13$, $14$, $16$, 
$17$, $18$, $20$, $23$, $24$, $25$, $27$, $28$, $29$, $31$, $35$, $39$, $46$,
$50$, $212$, $214$, $216$, $218$, $220$, $222$, $224$, $226$, $228$, $230$,
$232$, $233$, $234$, $235$, $236$, $237$, $238$, $239$, $240$, $241$, $242$,
$243$, $244$, $245$, $246$, $247$, $248$, $249$, $250$, $251$, $252$, $253$,
$254$, $255$, $256$, $257$, $258$, $259$, $261$, $262$, $263$, $264$, $265$,
$267$, $269$, $270$, $272$, $273$, $274$, $276$, $280$, $284$, $291$, $295$,
are the sum of at most $11$ base-$3$ digitally balanced 
numbers.
\label{base3}
\end{theorem}

\begin{proof}

\begin{observation}\label{obs:equi3}
     If $N$ is a positive number that is base-$3$ digitally balanced, then $27N+5$, $27N+7$, $27N+11$, $27N+15$, $27N+19$ and $27N+21$ are each positive base-3 digitally balanced numbers. This is because we have $5=(012)_3$, $7=(021)_3$, $11=(102)_3$, $15=(120)_3$, $19=(201)_3$, and $21=(210)_3$.
\end{observation}

We now show by brute force that some range of values can be written as the sum of 11 non-zero base-3 digitally balanced numbers, and then use induction to get 11 base-3 digitally balanced summands for values beyond this range.

\begin{lemma}\label{lem:11rep}
    For all $k\in\{0,1,2,\dots,26\}$ there exist integers $m_1,m_2,\dots,m_{11}\in\{5,7,11,15,19,21\}$ such that $\displaystyle k\equiv \sum_{i=1}^{11}m_i \pmod{27}$.
\end{lemma}
\begin{proof}
    Easy to verify, can fill in all 27 cases.
\end{proof}

\begin{lemma}\label{lem:geq622}
    Every natural number greater than or equal to $622$ is the sum of $11$ non-zero base-$3$ digitally balanced numbers.
\end{lemma}
\begin{proof}
    We can show by brute force checking that every integer $N$ satisfying $622\leq N\leq 17024$ can be written as the sum of 11 non-zero base-3 digitally balanced numbers. So assume that for all integers $m$ satisfying $622\leq m <N$ we have that $m$ can be written as the sum of 11 non-zero base-3 digitally balanced numbers.
    
    Suppose that $N\equiv k\pmod{27}$, for $k\in\{0,1,2,\dots,26\}$. Let $C=\sum_{i=1}^{11}c_i$ where $c_1,c_2,\dots,c_{11}\in\{5,7,11,15,19,21\}$ satisfy $k\equiv \sum_{i=1}^{11}c_i \pmod{27}$. Note such a set of values $c_1,c_2,\dots,c_{11}$ exists by Lemma \ref{lem:11rep}. Observe that $C\leq 11\times21=231$. Consider $m=\frac{N-C}{27}$. For $N>17024$ we have $m=\frac{N-C}{27}\geq \frac{N-231}{27}\geq \frac{17025-231}{27} \geq 622$ and $m=\frac{N-C}{27}<N$. So we can write $m=\sum_{i=1}^{11}m_i$ for non-zero base-3 digitally balanced numbers $m_1,m_2,\dots,m_{11}$ by the inductive hypothesis. Then by Observation~\ref{obs:equi3}, for each $1\leq i\leq 11$ we have $27m_i+c_i$ is a non-zero base-3 digitally balanced number. Thus we can write $$\sum_{i=1}^{11}27m_i + c_i=27\sum_{i=1}^{11}m_i + \sum_{i=1}^{11}c_i =27m +C=N.$$
    Therefore, by induction we have that every integer $N\geq 622$ can be written as the sum of 11 non-zero base-3 digitally balanced numbers. 
\end{proof}

Now we can complete the proof.

By Lemma \ref{lem:geq622} we have that every integer greater than or equal to 622 can be written as the sum of 11 base-3 digitally balanced numbers. Performing a brute force search on all natural numbers $N < 622$ yields a representation for $N$ as the sum of 11 base-3 digitally balanced numbers, except where $N$ is one of the values listed in the statement of the theorem. 
\end{proof}

We do not currently know if the bound $11$ is optimal.  By Theorem~\ref{bases} there is a lower bound of $10$.

\begin{remark}
The results of the last two sections can be viewed as special cases of a more general idea.
Suppose $S \subseteq \Enn$ is a set, with corresponding language $L$ of base-$k$ representations, and suppose
$L$ is closed under concatenation.
If $S$ has a finite subset $F$ with corresponding
language $L_F$ such that $L_F^* \not\subseteq w^*$ for all
words $w$,
then $L_F^*$ is a regular language of exponential density that underapproximates $L$.  If further
$\gcd([L_F^*]_k) = \gcd(S)$
then from \cite{Bell&Hare&Shallit:2017} 
it now follows that
there exists a constant $c$ such that every
sufficiently element of $S$ is a sum of at
most $c$ elements of $[L_F^*]_k$, and hence of $S$.
\end{remark}

\section{Limitations of the method}

It is natural to wonder whether more ``traditional'' problems in additive number theory can be handled by our technique.  
For example, suppose we try to approach Goldbach's conjecture (i.e., every even number $\geq 4$ is the sum of two primes) using a regular underapproximation of the language of primes in base $2$.
Unfortunately, this technique is guaranteed to fail, because a classical result of Hartmanis and Shank \cite{Hartmanis&Shank:1968} shows that every regular subset of the prime numbers is finite. 

Similarly, recent results on the additive properties of palindromes (discussed in 
Section~\ref{intro}) cannot be achieved by regular approximation, because
every regular language consisting solely of palindromes is {\it slender}:
it contains at most a constant
number of words of each length \cite{Horvath&Karhumaki&Kleijn:1987}.

We could also consider Waring's theorem, which concerns the additive properties of ordinary integer powers.  However, our approach also cannot work here, due to the following theorem.  Recall that the $k$-kernel of a set $S$ is defined to be the
number of distinct subsets of the form 
$$S_{e,j} = \{ n \in \Enn \ : \ k^e n + j \in S \text{ for } e \geq 0 \text{ and } 0 \leq j < k^e \}.$$

\begin{theorem} Let $k\geq 2$ be a natural number, and let $S$ be a $k$-automatic subset of $P:=\{n^j \colon n,j\ge 2\}$.  Then there is a finite set of integers $T$ such that $S\subseteq \{ c k^j \colon c\in T, j\ge 0\}$.  Moreover, if the size of the $k$-kernel of $S$ is $d$, then we can take $T$ to be a subset of $\{0,1,\ldots ,k^d-1\}$.
\label{bell-thm}
\end{theorem}

\begin{proof}
Let $S$ be a $k$-automatic subset of $P:=\{n^j \colon n,j\ge 2\}$.
We claim that for every natural number $m$, the set 
$$S^{(m)}:=\{n\in S\colon n\not\equiv \modd{0} {k^m} \}$$ 
is finite.  To see this, suppose that there is some $m$ such that $S^{(m)}$ is infinite.  Then since $S^{(m)}$ is $k$-automatic, by the pumping lemma it contains a set of elements of the form
$\{[xy^j  z]_k \colon j\ge 0\}$. Let $r$, $s$, and $t$ denote the lengths of $x$, $y$, and $z$ respectively.  We let $X=[x]_k$, $Y=[y]_k$, and $Z=[z]_k$. 
Then \begin{align*}
[xy^n z]_k &=
 Z + k^t(Y+ k^s Y +\cdots + k^{s(n-1)} Y) + k^{sn+t}X \\
 &= k^{sn}\left(k^t X + \frac{k^t Y}{k^s-1}\right)+\left(Z-\frac{k^t Y}{k^s-1}\right).\end{align*}
 Then $u_n:=[x y^n z]$ satisfies the two-term linear recurrence
 $u_n = r_1 u_{n-1} + r_2 u_{n-2}$ with 
 $r_1= (k^s+1)u_{n-1}$ and $r_2= -k^s u_{n-1}$.  In particular, $r_1^2+4r_2^2\neq 0$, and since $k^s$ and $1$ are nonzero and not roots of unity, we have that the recurrence is 
 non-degenerate as long as 
 $k^t X + Y k^t/(k^s-1)$ and $Z-k^t Y/(k^s-1)$ are nonzero.  But since $u_n=[x y^n z]_k\to \infty$ as $n\to\infty$, we see that
 $k^t X + k^t Y/(k^s-1)$ must be nonzero; since $u_n\not \equiv \modd{0} {k^m}$, we see that $Z-k^t Y/(k^s-1)$ is nonzero.  Then, by \cite[Theorem 2]{SS} we deduce that $P\cap \{u_n\colon n\ge 0\}$ is finite, a contradiction.  It follows that $S^{(m)}$ is finite for every $m\ge 1$.  
 
We now finish the proof. Let $d$ denote the size of the $k$-kernel of $S$ and let $T=S^{(d)}\cup \{0\}$.  Then $T$ is a finite set. We claim that $S\subseteq S_0:=\{ c k^j \colon c\in T, j\ge 0\}$.  To see this, suppose that this is not the case. Then there is some $\ell \in S\setminus S_0$.  Pick the smallest natural number $\ell  \not\in S \setminus S_0$.  Since $0\in T$, we have that $\ell$ is positive.  Also since $T\subseteq S_0$,  we have $\ell\not\in T$, and so $\ell$ must be divisible by $k^d$ (since if it were not, it would be in $T$).  Thus $k^d \mid \ell$.  The $k$-kernel of $S$ has size $d$, and so there exist $i,j\in \{0,1,\ldots ,d\}$ with $i<j$ such that 
$$\{n\in \mathbb{N}\colon k^i n\in S\} = \{n\in \mathbb{N}\colon k^j n \in S\}.$$  
In particular, if $k^j$ divides $n$ and $k^j n$ is in $S$, then $k^{i-j}n\in S$.  Then, since $k^d \mid \ell$, we see that $k^{i-j}\ell \in S$; furthermore $k^{i-j}\ell < \ell$ since $\ell$ is positive.  Thus, by minimality of $\ell$, we have $k^{i-j}\ell\in S_0$, and since $S_0$ is closed under multiplication by $k$, we get that $\ell=k^{i-j}\ell k^{j-i}$ is also in $S_0$, a contradiction.  The result follows. 
\end{proof}

\end{document}